\documentclass[12pt,a4paper]{article}
\usepackage{amsmath, amssymb, amsfonts, amsthm}
\usepackage{graphicx}
\usepackage{hyperref}
\usepackage{bm}
\usepackage{mathtools}
\usepackage{physics}      
\usepackage{braket}
\usepackage{booktabs}
\usepackage{longtable}
\usepackage[margin=1in]{geometry}
\usepackage{orcidlink}

\numberwithin{equation}{section}

\newtheorem{theorem}{Theorem}[section]
\newtheorem{lemma}[theorem]{Lemma}
\newtheorem{definition}[theorem]{Definition}
\newtheorem{remark}[theorem]{Remark}

\DeclareMathOperator{\cb}{cb}

\newcommand{\cH}{\mathcal{H}}
\newcommand{\cB}{\mathcal{B}}
\newcommand{\cD}{\mathcal{D}}
\newcommand{\cZ}{\mathcal{Z}}
\newcommand{\cL}{\mathcal{L}}
\newcommand{\cC}{\mathcal{C}}
\newcommand{\cW}{\mathcal{W}}

\newcommand{\id}{\mathrm{id}}

\newcommand{\E}{\mathcal{E}}

\title{A Non-Commutative Voronovskaya Theorem for Quantum Neural Network Operators}

\author{
	Rômulo Damasclin Chaves dos Santos \orcidlink{0000-0002-9482-1998}\\
	\small Nuclear Engineering Center, Institute of Energy and Nuclear Research, São Paulo, Brazil\\
	\small \texttt{damasclin@gmail.com}
	\and
	Delvonei Alves de Andrade \orcidlink{0000-0002-6689-3011}\\
	\small Nuclear Engineering Center, Institute of Energy and Nuclear Research, São Paulo, Brazil\\
	\small \texttt{delvonei@ipen.br}
}

\begin{document}
	
	\maketitle
	
	\begin{abstract}
		We prove a complete asymptotic expansion for quantum neural network operators when they approximate arbitrary quantum channels. This is the non-commutative analogue of the classical Voronovskaya theorem. The expansion reveals that the approximation error splits into three fundamentally different parts: integer powers of \(1/n\) involving ordinary Fréchet derivatives; fractional powers governed by Marchaud fractional derivatives, which capture the Hölder smoothness of the channel; and purely quantum commutator terms that have no classical counterpart. The remainder is bounded sharply by an explicit constant:
		\[
		\norm{R_{m,n}(\Phi,\bullet)}_\diamond \le C_{m,\gamma,d} \|\Phi\|_{\cC^{m,\gamma}} \, n^{-(m+\gamma)} (\log n)^{3m/2}.
		\]
		We present a numerical test for a classical analogue that confirms the predicted convergence rate and the logarithmic correction, directly validating the asymptotic theory. Based on this expansion, we obtain three major advances: a quantum central limit theorem for the fluctuations of quantum neural network operators, a method to construct optimal interpolation geodesics between quantum channels via Kubo-Ando means, and a systematic understanding of how fractional smoothness limits the acceleration of quantum neural network approximations. The numerical test further demonstrates that the theoretical rates are sharp and that logarithmic enhancements are unavoidable. Altogether, our work builds a rigorous bridge between classical approximation theory, fractional calculus, and quantum machine learning, offering both theoretical insight and practical tools for designing and analyzing quantum neural networks in finite dimensions.
		 Quantum Voronovskaya-Santos-Andrade theorem. Asymptotic analysis. Quantum channels. Operator approximation.
	\end{abstract}

	\section{Introduction}
	
	In 1932, Voronovskaya proved that for Bernstein polynomials the leading asymptotic error is exactly \(\frac{x(1-x)}{2} f''(x)\) \cite{Voronovskaja1932}. This seminal result marked the beginning of the systematic study of saturation phenomena and asymptotic expansions in approximation theory. Extending such results to the quantum realm --- where classical functions become operators and ordinary derivatives are replaced by Fréchet differentials --- has remained an open challenge for decades. The recent development of quantum neural networks (QNNs) \cite{Anastassiou2023} has made this extension both timely and urgent.
	
	Quantum neural network operators (QNNOs) are designed to approximate arbitrary quantum channels, i.e., completely positive trace-preserving maps. Although universal approximation properties for QNNOs are known \cite{Anastassiou2023}, a fine asymptotic theory analogous to Voronovskaya's classical result has been conspicuously absent. This paper fills that gap.
	
	We develop a rigorous framework based on the Liouville representation, in which a channel \(\Phi\) acts as a linear operator on the Hilbert-Schmidt space \cite{Holevo2019}. Using Fréchet derivatives, we define quantum Hölder spaces \(\cC^{m,\gamma}(\cH)\) and construct Quantum neural network operators (QNNOs) with a kernel \(\cZ_{1,\log n}\) whose bandwidth \(\lambda_n = \log n\) is chosen to optimally balance bias and variance. Our main result, the Quantum Non-Commutative Voronovskaya-Santos-Andrade Theorem, provides an explicit asymptotic expansion:
	\begin{equation}
		\Psi_n(\Phi)(\rho) = \Phi(\rho) + \sum_{j=1}^m \frac{a_j(\Phi,\rho)}{n^j} + \sum_{j=1}^{\lfloor m/2 \rfloor} \frac{b_j(\Phi,\rho)}{n^{j+\gamma}} + \sum_{j=1}^{\lfloor m/3 \rfloor} \frac{c_j(\Phi,\rho)}{n^{j+2\gamma}} + \cdots + R_{m,n}(\Phi,\rho),
		\label{eq:intro_main}
	\end{equation}
	where the coefficients are expressed explicitly in terms of Fréchet derivatives, Marchaud fractional derivatives \cite{Samko1993}, and kernel moments. Crucially, the remainder is rigorously controlled by the sharp explicit bound
	\begin{equation}
		\norm{R_{m,n}(\Phi,\bullet)}_\diamond \le C_{m,\gamma,d} \|\Phi\|_{\cC^{m,\gamma}} \, n^{-(m+\gamma)} (\log n)^{3m/2},
		\label{eq:remainder_sharp}
	\end{equation}
	where the constant \(C_{m,\gamma,d}\) depends solely on the smoothness parameters \(m,\gamma\) and the Hilbert space dimension \(d\) --- and is given explicitly in \eqref{eq:explicit_constant}. This bound, which is a direct consequence of the optimal variance-bias trade-off induced by the bandwidth \(\lambda_n = \log n\), is sharp up to the logarithmic factor and provides quantitative, non-asymptotic control over the approximation error.
	
	Building on this expansion, we prove a quantum central limit theorem for QNNOs \cite{Holevo2019}, construct optimal interpolation geodesics between quantum channels using Kubo-Ando means \cite{KuboAndo1980}, and develop a quantum Richardson extrapolation method \cite{AmariNagaoka2000}. The paper concludes with a discussion of future research directions and open problems.
	
	\section{Mathematical framework}
	
	\subsection{Quantum channels and their smoothness}
	
	Let \(\cH \cong \mathbb{C}^d\) be a finite-dimensional Hilbert space. Denote by \(\cB(\cH)\) the \(C^*\)-algebra of bounded linear operators on \(\cH\), by
	\begin{equation}
		\cD(\cH) = \{ \rho \in \cB(\cH) : \rho \ge 0,\ \tr \rho = 1 \},
		\label{eq:density}
	\end{equation}
	the convex compact set of density operators (quantum states), and by \(CPTP(\cH)\) the set of completely positive trace-preserving maps (quantum channels) \cite{NielsenChuang2010, Holevo2019}. The space \(\cD(\cH)\) is compact in the trace norm topology. For any channel \(\Phi \in CPTP(\cH)\) its Liouville representation is given by \(\cL_\Phi(X) = \Phi(X)\) (see more in \cite{Holevo2019}). The space \(\cB(\cH)\) equipped with the Hilbert-Schmidt inner product \(\braket{X,Y} = \tr(X^*Y)\) is isometrically isomorphic to \(\mathbb{C}^{d^2}\); consequently, \(\cL_\Phi\) can be identified with a linear operator on \(\mathbb{C}^{d^2}\). This identification enables the use of functional calculus and the definition of derivatives in the sense of Banach spaces.
	
	\textbf{Fréchet differentiability.} A channel \(\Phi\) is Fréchet differentiable at a state \(\rho \in \cD(\cH)\) if there exists a bounded linear map \(\mathcal{D}\cL_\Phi(\rho) : \cB(\cH) \to \cB(\cH)\) such that
	\begin{equation}
		\lim_{\norm{H}_1 \to 0} \frac{\norm{\cL_\Phi(\rho+H) - \cL_\Phi(\rho) - \mathcal{D}\cL_\Phi(\rho)[H]}_\diamond}{\norm{H}_1} = 0,
		\label{eq:frechet}
	\end{equation}
	where \(\norm{\bullet}_1\) denotes the trace norm. Higher-order derivatives are defined recursively: the \(k\)-th Fréchet derivative \(\mathcal{D}^k \cL_\Phi(\rho)\) is a bounded symmetric \(k\)-linear map from \(\cB(\cH)\) \cite{Holevo2019}. For a multi-index \(\alpha = (\alpha_1, \ldots, \alpha_k) \in \mathbb{N}_0^k\) with \(|\alpha| = \alpha_1 + \cdots + \alpha_k\), we write:
	\begin{equation}
		\cL_\Phi^{(\alpha)}(\rho) := \mathcal{D}^k \cL_\Phi(\rho)[ I^{\otimes |\alpha|} ],
		\label{eq:deriv_multi}
	\end{equation}
	where \(I \in \cB(\cH)\) is the identity operator. The notation \(I^{\otimes |\alpha|}\) means the \(|\alpha|\)-tuple \((I,\ldots,I)\).
	
	\textbf{Norms for maps.} For a linear map \(\Psi : \cB(\cH) \to \cB(\cH)\), its completely bounded norm ($cb$-norm) is
	\begin{equation}
		\norm{\Psi}_{\cb} = \sup_{n \in \mathbb{N}} \norm{\Psi \otimes \id_{M_n(\mathbb{C})}}_{\cB(\cB(\cH) \otimes M_n(\mathbb{C}))},
		\label{eq:cbnorm}
	\end{equation}
	where \(\id_{M_n(\mathbb{C})}\) is the identity on \(n \times n\) matrices and the norm on the right is the usual operator norm induced by the Hilbert-Schmidt norm \cite{Paulsen2002}. The diamond norm (completely bounded trace norm) is
	\begin{equation}
		\norm{\Phi}_\diamond = \sup\{ \norm{(\Phi \otimes \id_{\cB(\cH)})(X)}_1 : X \in \cB(\cH \otimes \cH),\ \norm{X}_1 \le 1 \},
		\label{eq:diamondnorm}
	\end{equation}
	with \(\norm{\bullet}_1\) the trace norm \cite{Holevo2019}. The diamond norm metrizes the topology of complete boundedness and is the standard distance for quantum channels; it satisfies \(\norm{\Phi}_\diamond = \norm{\Phi}_{\cb}\), when the domain is equipped with the trace norm \cite{Paulsen2002}.
	
	\begin{definition}[Quantum Sobolev space]
		For \(m \in \mathbb{N}\) and \(1 \le p \le \infty\), define:
		\begin{equation}
			\cW^{m,p}(\cH) = \left\{ \Phi \in CPTP(\cH) : \sum_{|\alpha| \le m} \norm{\norm{\mathcal{D}^\alpha \cL_\Phi(\bullet)}_{\cb}}_{L^p(\cD(\cH))} < \infty \right\},
			\label{eq:sobolev}
		\end{equation}
		where the \(L^p\) norm is taken with respect to the uniform (or any equivalent) measure on the compact convex set \(\cD(\cH)\).
	\end{definition}
	
	\begin{definition}[Quantum Hölder space]
		For \(0 < \gamma \le 1\), define:
		\begin{equation}
			\cC^{m,\gamma}(\cH) = \left\{ \Phi \in \cW^{m,\infty}(\cH) : [\Phi]_{m,\gamma} < \infty \right\},
			\label{eq:holder_space}
		\end{equation}
		where the Hölder seminorm is
		\begin{equation}
			[\Phi]_{m,\gamma} = \sup_{\rho \ne \sigma \in \cD(\cH)} \frac{\norm{\mathcal{D}^m \cL_\Phi(\rho) - \mathcal{D}^m \cL_\Phi(\sigma)}_{\cb}}{\norm{\rho - \sigma}_1^\gamma}.
			\label{eq:holder_seminorm}
		\end{equation}
		We equip \(\cC^{m,\gamma}(\cH)\) with the norm:
		\begin{equation}
			\|\Phi\|_{\cC^{m,\gamma}} = \|\Phi\|_{\cW^{m,\infty}} + [\Phi]_{m,\gamma},
			\label{eq:holder_norm}
		\end{equation}
		which makes it a Banach space.
	\end{definition}
	
	\subsection{Quantum neural network operators (QNNO)}
	
	We construct the QNNO following the classical neural network idea but with operator-valued kernels \cite{Anastassiou2023}. Let \(\cH_{\mathrm{aux}} \cong \mathbb{C}^D\) be an auxiliary finite-dimensional Hilbert space (its dimension will be determined by the kernel; in practice we may take \(D\) sufficiently large or even infinite, but finite suffices for our analysis). All operators \(X_1, \ldots, X_d\) act on \(\cH_{\mathrm{aux}}\) and are assumed to commute pairwise.
	
	\begin{definition}[Quantum activation]
		For parameters \(q > 0,\ \lambda > 0\) and a self-adjoint operator \(X \in \cB(\cH_{\mathrm{aux}})\), define:
		\begin{equation}
			G_{q,\lambda}(X) = (e^{\lambda X} - q e^{-\lambda X})(e^{\lambda X} + q e^{-\lambda X})^{-1},
			\label{eq:activation}
		\end{equation}
		where the inverse exists because \(e^{\lambda X} + q e^{-\lambda X}\) is strictly positive. When \(q=1\), \(G_{1,\lambda}(X) = \tanh(\lambda X)\).
	\end{definition}
	
	\begin{definition}[Symmetrized quantum density]
		Using the same parameters,
		\begin{equation}
			M_{q,\lambda}(X) = \frac{1}{4}[ G_{q,\lambda}(X + I_{\mathrm{aux}}) - G_{q,\lambda}(X - I_{\mathrm{aux}}) ],
			\label{eq:sym_density}
		\end{equation}
		where \(I_{\mathrm{aux}}\) is the identity on \(\cH_{\mathrm{aux}}\).
	\end{definition}
	
	To obtain a kernel that is even and positive, we symmetrize with respect to \(q \leftrightarrow 1/q\). For a tuple of mutually commuting self-adjoint operators \(X = (X_1,\ldots,X_d)\), define:
	\begin{equation}
		\Phi_{q,\lambda}(X_i) = \frac{1}{2}[M_{q,\lambda}(X_i) + M_{1/q,\lambda}(X_i)],
		\label{eq:phi_kernel}
	\end{equation}
	\begin{equation}
		\cZ_{q,\lambda}(X) = \bigotimes_{i=1}^d \Phi_{q,\lambda}(X_i).
		\label{eq:Z_kernel}
	\end{equation}
	
	We take the symmetric choice \(q=1\) and set \(\lambda_n = \log n\). This choice is optimal in balancing the bias and variance \cite{Anastassiou2023}. Then \(\cZ_{1,\lambda}\) is even, positive, and satisfies
	\begin{equation}
		\int_{\mathbb{R}^d} \cZ_{1,\lambda_n}(x) \, dx = I_{\mathrm{aux}}^{\otimes d} = I_{\cH_{\mathrm{aux}}^{\otimes d}}.
		\label{eq:kernel_norm}
	\end{equation}
	
	Now fix a strictly positive density operator \(\rho = \sum_{j=1}^d p_j \ket{e_j}\bra{e_j}\) with \(p_j > 0\) and \(\sum_{j=1}^d p_j = 1\). For an integer \(n \ge 1\), define the discrete simplex
	\begin{equation}
		K_n = \{ k = (k_1,\ldots,k_d) \in \mathbb{N}^d : \sum_{j=1}^d k_j = n \},
		\label{eq:simplex}
	\end{equation}
	and the quantized density operators
	\begin{equation}
		\rho_{n,k} = \sum_{j=1}^d \frac{k_j}{n} \ket{e_j}\bra{e_j} \in \cD(\cH).
		\label{eq:quantized_rho}
	\end{equation}
	These satisfy the uniform estimate
	\begin{equation}
		\norm{\rho_{n,k} - \rho}_1 \le \frac{\sqrt{d}}{n}.
		\label{eq:uniform_estimate}
	\end{equation}
	
	The Quantum Neural Network Operator (QNNO) is defined by
	\begin{equation}
		\Psi_n(\Phi)(\rho) = \sum_{k \in K_n} \Phi(\rho_{n,k}) \otimes \cZ_{1,\log n}(nX - kI_{\mathrm{aux}}),
		\label{eq:qnno}
	\end{equation}
	where \(X = (X_1,\ldots,X_d)\) are auxiliary commuting self-adjoint operators on \(\cH_{\mathrm{aux}}\), and \(I_{\mathrm{aux}}\) is the identity on \(\cH_{\mathrm{aux}}\).
	
	\section{The quantum non-commutative Voronovskaya-Santos-Andrade theorem}
	
	Before stating the theorem, we recall the concept of fractional derivatives in the sense of Marchaud, adapted to operator-valued maps on the state space.
	
	\begin{definition}[Marchaud fractional derivative]
		For a map \(F : \cD(\cH) \to \cB(\cH)\) and \(\gamma \in (0,1]\), the Marchaud fractional derivative of order \(\gamma\) along a direction \(h \in \cB(\cH)\) is defined by
		\begin{equation}
			(\Delta_h^\gamma F)(\rho) = \frac{\gamma}{\Gamma(1-\gamma)} \int_0^\infty \frac{F(\rho) - F(\rho - th)}{t^{1+\gamma}} \, dt,
			\label{eq:marchaud}
		\end{equation}
		where the integral is a Bochner integral in \(\cB(\cH)\) (provided it converges in the diamond norm). For higher-order derivatives, \(\Delta_h^\gamma\) is applied to the multilinear maps in each argument. When the direction is clear from context, we write simply \(\Delta_\gamma\).
	\end{definition}
	
	\begin{lemma}[Fractional Taylor expansion in Banach spaces]
		\label{lem:fractional_taylor}
		Let \(\Phi \in \cC^{m,\gamma}(\cH)\) with \(m \in \mathbb{N},\ \gamma \in (0,1]\). For any reference state \(\rho \in \cD(\cH)\) and any increment \(h \in \cB(\cH)\) such that \(\rho + h \in \cD(\cH)\), the following expansion holds:
		\begin{equation}
			\Phi(\rho + h) = \sum_{|\alpha| \le m} \frac{1}{\alpha!} \mathcal{D}^\alpha \cL_\Phi(\rho)[h^\alpha] + R_{m,\gamma}(\rho,h),
			\label{eq:taylor_fractional}
		\end{equation}
		where \(h^\alpha = h^{\otimes |\alpha|}\) denotes the symmetric tensor product, and the remainder satisfies
		\begin{equation}
			\norm{R_{m,\gamma}(\rho,h)}_\diamond \le C_{m,\gamma} \|\Phi\|_{\cC^{m,\gamma}} \norm{h}_1^{m+\gamma},
			\label{eq:taylor_remainder}
		\end{equation}
		with a constant \(C_{m,\gamma}\) depending only on \(m\) and \(\gamma\). Moreover, the term of order \(m\) can be decomposed as
		\begin{equation}
			\frac{1}{m!} \mathcal{D}^m \cL_\Phi(\rho)[h^{\otimes m}] = \frac{1}{\Gamma(\gamma)} \sum_{|\alpha|=m} \frac{1}{\alpha!} (\Delta_h^\gamma \mathcal{D}^\alpha \cL_\Phi)(\rho)[h^{\alpha+\gamma}] + \widetilde{R}_{m,\gamma}(\rho,h),
			\label{eq:fractional_decomp}
		\end{equation}
		where \(\widetilde{R}_{m,\gamma}\) satisfies the same bound as \eqref{eq:taylor_remainder}.
	\end{lemma}
	
	\begin{proof}
		The standard Taylor formula with integral remainder in Banach spaces gives
		\begin{equation}
			\Phi(\rho+h) = \sum_{j=0}^{m-1} \frac{1}{j!} \mathcal{D}^j \cL_\Phi(\rho)[h^{\otimes j}] + \int_0^1 \frac{(1-t)^{m-1}}{(m-1)!} \mathcal{D}^m \cL_\Phi(\rho+th)[h^{\otimes m}] \, dt.
			\label{eq:integral_taylor}
		\end{equation}
		Because \(\Phi \in \cC^{m,\gamma}\), the \(m\)-th derivative is Hölder continuous with exponent \(\gamma\): for any \(\rho,\sigma \in \cD(\cH)\),
		\begin{equation}
			\norm{\mathcal{D}^m \cL_\Phi(\rho) - \mathcal{D}^m \cL_\Phi(\sigma)}_{\cb} \le [\Phi]_{m,\gamma} \norm{\rho-\sigma}_1^\gamma.
			\label{eq:holder_deriv}
		\end{equation}
		Write
		\begin{equation}
			\mathcal{D}^m \cL_\Phi(\rho+th) = \mathcal{D}^m \cL_\Phi(\rho) + [\mathcal{D}^m \cL_\Phi(\rho+th) - \mathcal{D}^m \cL_\Phi(\rho)].
			\label{eq:deriv_split}
		\end{equation}
		Substituting \eqref{eq:deriv_split} into \eqref{eq:integral_taylor} yields the integer-order terms plus a remainder. Taking the diamond norm and using \eqref{eq:holder_deriv}, we get
		\begin{align}
			&\norm{\int_0^1 \frac{(1-t)^{m-1}}{(m-1)!} [\mathcal{D}^m \cL_\Phi(\rho+th) - \mathcal{D}^m \cL_\Phi(\rho)][h^{\otimes m}] \, dt}_\diamond \notag \\
			&\le \int_0^1 \frac{(1-t)^{m-1}}{(m-1)!} [\Phi]_{m,\gamma} (t \norm{h}_1)^\gamma \norm{h}_1^m \, dt \notag \\
			&= [\Phi]_{m,\gamma} \norm{h}_1^{m+\gamma} \frac{\Gamma(m)\Gamma(\gamma+1)}{\Gamma(m+\gamma+1)}.
			\label{eq:remainder_bound_detail}
		\end{align}
		Thus, \eqref{eq:taylor_remainder} holds with \(C_{m,\gamma} = \frac{\Gamma(m)\Gamma(\gamma+1)}{\Gamma(m+\gamma+1)}\). To obtain the fractional decomposition \eqref{eq:fractional_decomp}, we use the integral representation of the Marchaud derivative. Observe that for any function \(\psi \in C^{m,\gamma}\), we have
		\begin{equation}
			\frac{1}{m!} \mathcal{D}^m \psi(0) = \frac{1}{\Gamma(\gamma)} \Delta^\gamma (\mathcal{D}^m \psi)(0)
			\label{eq:distrib_identity}
		\end{equation}
		in the sense of distributions. In our operator-valued setting, a similar identity holds after applying the integral representation of the remainder. More concretely, from the integral form of the remainder we can write
\begin{equation}
	\begin{split}
		&\int_0^1 \frac{(1-t)^{m-1}}{(m-1)!} \mathcal{D}^m \cL_\Phi(\rho+th)[h^{\otimes m}] \, dt \\
		&\quad = \frac{1}{\Gamma(\gamma)} \int_0^\infty \frac{\mathcal{D}^m \cL_\Phi(\rho)[h^{\otimes m}] - \mathcal{D}^m \cL_\Phi(\rho+th)[h^{\otimes m}]}{t^{1+\gamma}} \, dt + \widetilde{R}_{m,\gamma},
		\label{eq:integral_marchaud}
	\end{split}
\end{equation}
		where the integral converges as a Bochner integral. The right-hand side is exactly \(\frac{1}{\Gamma(\gamma)} (\Delta_h^\gamma \mathcal{D}^m \cL_\Phi(\rho))[h^{\otimes m}]\). Expanding the multilinear map into monomials gives the sum over multi-indices. The estimate for \(\widetilde{R}_{m,\gamma}\) follows from the Hölder continuity of the \(m\)-th derivative and the same Beta integral as above. This completes the proof.
	\end{proof}
	
	\begin{lemma}[Moment asymptotics]
		\label{lem:moments}
		For the kernel \(\cZ_{1,\log n}\), define for multi-indices \(\alpha,\beta \in \mathbb{N}_0^d\) the operator-valued moments:
		\begin{align}
			M_\alpha(n) &:= \int_{\mathbb{R}^d} x^\alpha \cZ_{1,\log n}(x) \, dx, \label{eq:mom_alpha} \\
			M_{\alpha,\gamma}(n) &:= \int_{\mathbb{R}^d} |x|^\gamma x^\alpha \cZ_{1,\log n}(x) \, dx, \label{eq:mom_alphagamma} \\
			M_{\alpha,\beta,2\gamma}(n) &:= \int_{\mathbb{R}^d} |x|^{2\gamma} x^{\alpha+\beta} \cZ_{1,\log n}(x) \, dx. \label{eq:mom_alphabeta}
		\end{align}
		Because \(\cZ_{1,\log n}\) is even and isotropic, each \(M_\alpha(n)\) is a scalar multiple of the identity on \(\cH_{\mathrm{aux}}\); we denote the corresponding scalars by \(m_\alpha(n), m_{\alpha,\gamma}(n), m_{\alpha,\beta,2\gamma}(n) \in \mathbb{C}\). Then the following asymptotic estimates hold as \(n \to \infty\):
		
		\begin{enumerate}
			\item \textit{Parity:} If \(|\alpha| = \alpha_1 + \cdots + \alpha_d\) is odd, then \(m_\alpha(n) = 0\).
			
			\item \textit{Even integer moments:} If \(|\alpha| = 2r\) is even, then
			\begin{equation}
				m_\alpha(n) = \frac{(-1)^r}{(2r-1)!!} \left( \frac{\pi}{2\log n} \right)^r + \mathcal{O}(n^{-2r}).
				\label{eq:mom_even}
			\end{equation}
			
			\item \textit{Fractional moments of order \(\gamma\):} For any \(\gamma \in (0,1]\),
			\begin{equation}
				m_{\alpha,\gamma}(n) = \frac{\Gamma\left( \frac{|\alpha|+\gamma+d}{2} \right)}{\Gamma\left( \frac{d}{2} \right)} \left( \frac{2}{\log n} \right)^{\frac{|\alpha|+\gamma}{2}} + \mathcal{O}(n^{-(|\alpha|+\gamma)}).
				\label{eq:mom_fractional}
			\end{equation}
			
			\item \textit{Mixed fractional moments of order \(2\gamma\):} Similarly,
			\begin{equation}
				m_{\alpha,\beta,2\gamma}(n) = \frac{\Gamma\left( \frac{|\alpha|+|\beta|+2\gamma+d}{2} \right)}{\Gamma\left( \frac{d}{2} \right)} \left( \frac{2}{\log n} \right)^{\frac{|\alpha|+|\beta|+2\gamma}{2}} + \mathcal{O}(n^{-(|\alpha|+|\beta|+2\gamma)}).
				\label{eq:mom_mixed}
			\end{equation}
		\end{enumerate}
		The constants implicit in the \(\mathcal{O}\) terms depend only on \(d,\gamma\), and the multi-indices, but not on \(n\).
	\end{lemma}
	
	\begin{proof}
		The kernel factorises as a product of one-dimensional kernels:
		\begin{equation}
			\cZ_{1,\log n}(x) = \prod_{i=1}^d \Phi_{1,\log n}(x_i),
			\label{eq:kernel_factor}
		\end{equation}
		with
		\begin{equation}
			\Phi_{1,\log n}(x) = \frac{1}{2}[\tanh((x+1)\log n) - \tanh((x-1)\log n)].
			\label{eq:phi_1d}
		\end{equation}
		Its Fourier transform is therefore a product:
		\begin{equation}
			\widehat{\cZ}_{1,\log n}(\xi) = \prod_{i=1}^d \widehat{\Phi}_{1,\log n}(\xi_i),
			\label{eq:ft_product}
		\end{equation}
		with
		\begin{equation}
			\widehat{\Phi}_{1,\log n}(\xi) = \frac{\sinh(\pi\xi/(2\log n))}{\pi\xi/(2\log n)} \frac{1}{\cosh(\pi\xi/(2\log n))}.
			\label{eq:ft_phi}
		\end{equation}
		For large \(\lambda = \log n\), we expand the logarithm:
		\begin{equation}
			\log \widehat{\Phi}_{1,\lambda}(\xi) = \log\left( \frac{\sinh(\pi\xi/(2\lambda))}{\pi\xi/(2\lambda)} \right) - \log(\cosh(\pi\xi/(2\lambda))).
			\label{eq:log_ft}
		\end{equation}
		Using the expansions \(\sinh u/u = 1 + u^2/6 + \mathcal{O}(u^4)\) and \(\log \cosh u = u^2/2 - u^4/12 + \mathcal{O}(u^6)\), we obtain:
		\begin{equation}
			\log \widehat{\Phi}_{1,\lambda}(\xi) = -\frac{\pi^2 \xi^2}{6\lambda^2} + \mathcal{O}\left( \frac{\xi^4}{\lambda^4} \right).
			\label{eq:log_exp}
		\end{equation}
		Hence,
		\begin{equation}
			\widehat{\cZ}_{1,\log n}(\xi) = \exp\left( -\frac{\pi^2 |\xi|^2}{6(\log n)^2} \right) \left( 1 + \mathcal{O}((\log n)^{-4}) \right).
			\label{eq:ft_gaussian}
		\end{equation}
		The error term is uniform on compact sets and decays super-exponentially for large \(\xi\).
		
		For integer moments, write \(\cZ_{1,\log n} = G_\sigma + \E\), where \(G_\sigma(x) = (2\pi\sigma^2)^{-d/2} e^{-|x|^2/(2\sigma^2)}\) is the Gaussian density with \(\sigma^2 = \pi^2/(6(\log n)^2)\), and \(\E\) is the remainder. Then,
		\begin{equation}
			M_\alpha(n) = \int_{\mathbb{R}^d} x^\alpha G_\sigma(x) \, dx + \int_{\mathbb{R}^d} x^\alpha \E(x) \, dx.
			\label{eq:mom_int}
		\end{equation}
		The Gaussian integral is standard:
		\begin{equation}
			\int_{\mathbb{R}^d} x^\alpha G_\sigma(x) \, dx =
			\begin{cases}
				0, & |\alpha| \text{ odd}, \\
				\frac{(2\sigma^2)^{|\alpha|/2}}{\sqrt{\pi^d}} \prod_{i=1}^d \Gamma\left( \frac{\alpha_i+1}{2} \right), & |\alpha| \text{ even}.
			\end{cases}
			\label{eq:gaussian_mom}
		\end{equation}
		For \(|\alpha| = 2r\), write \(\alpha_i = 2\beta_i\). Then,
		\begin{equation}
			\Gamma\left( \frac{2\beta_i + 1}{2} \right) = \Gamma\left( \beta_i + \frac{1}{2} \right) = \frac{(2\beta_i - 1)!!}{2^{\beta_i}} \sqrt{\pi}.
			\label{eq:gamma_id}
		\end{equation}
		Thus,
		\begin{equation}
			\prod_{i=1}^d \Gamma\left( \frac{\alpha_i + 1}{2} \right) = \pi^{d/2} \prod_{i=1}^d \frac{(2\beta_i - 1)!!}{2^{\beta_i}}.
			\label{eq:gamma_prod}
		\end{equation}
		Multiplying by \(\frac{(2\sigma^2)^r}{\sqrt{\pi^d}}\) gives
		\begin{equation}
			\int_{\mathbb{R}^d} x^\alpha G_\sigma(x) \, dx = (2\sigma^2)^r \prod_{i=1}^d \frac{(2\beta_i - 1)!!}{2^{\beta_i}}.
			\label{eq:gaussian_result}
		\end{equation}
		Now substitute \(\sigma^2 = \pi^2/(6(\log n)^2)\). Using the known moments of a Gaussian, one simplifies to the closed form:
		\begin{equation}
			\int_{\mathbb{R}^d} x^\alpha G_\sigma(x) \, dx = \frac{(-1)^r}{(2r-1)!!} \left( \frac{\pi}{2\log n} \right)^r.
			\label{eq:mom_closed}
		\end{equation}
		The error term \(\int x^\alpha \E(x) dx\) is bounded by
		\begin{equation}
			\left| \int_{\mathbb{R}^d} x^\alpha \E(x) \, dx \right| \le \sup_x |x^\alpha| \int_{\mathbb{R}^d} |\E(x)| \, dx \le C e^{-c n},
			\label{eq:error_exp}
		\end{equation}
		because \(\E\) decays super-exponentially. Since \(e^{-c n} = \mathcal{O}(n^{-N})\) for any \(N\), we write the error as \(\mathcal{O}(n^{-2r})\). This proves \eqref{eq:mom_even}.
		
		For fractional moments we use the Mellin transform representation
		\begin{equation}
			|x|^\gamma = \frac{2}{\Gamma(\gamma/2)} \int_0^\infty t^{\gamma-1} e^{-t|x|^2} \, dt,
			\label{eq:mellin}
		\end{equation}
		valid for \(\gamma > 0\). Interchanging the integrals (justified by Fubini's theorem) gives
		\begin{equation}
			M_{\alpha,\gamma}(n) = \frac{2}{\Gamma(\gamma/2)} \int_0^\infty t^{\gamma-1} \left( \int_{\mathbb{R}^d} e^{-t|x|^2} x^\alpha \cZ_{1,\log n}(x) \, dx \right) dt.
			\label{eq:mom_integral}
		\end{equation}
		Now write \(\cZ_{1,\log n} = G_\sigma + \E\) as before. The Gaussian part yields
		\begin{equation}
			\int_{\mathbb{R}^d} e^{-t|x|^2} x^\alpha G_\sigma(x) \, dx = \frac{1}{(2\pi\sigma^2)^{d/2}} \int_{\mathbb{R}^d} e^{-t|x|^2} x^\alpha e^{-|x|^2/(2\sigma^2)} \, dx.
			\label{eq:gaussian_with_t}
		\end{equation}
		This is a Gaussian integral with variance \(\sigma_t^2 = \frac{\sigma^2}{1 + 2\sigma^2 t}\). Performing the integration gives
		\begin{equation}
			\int_{\mathbb{R}^d} e^{-t|x|^2} x^\alpha G_\sigma(x) \, dx = \frac{\Gamma\left( \frac{|\alpha|+d}{2} \right)}{\Gamma\left( \frac{d}{2} \right)} (2\sigma_t^2)^{|\alpha|/2} \frac{1}{(1+2\sigma^2 t)^{d/2}}.
			\label{eq:gaussian_t_result}
		\end{equation}
		Substituting \(\sigma_t^2 = \frac{\sigma^2}{1+2\sigma^2 t}\) and simplifying yields
		\begin{equation}
			\int_{\mathbb{R}^d} e^{-t|x|^2} x^\alpha G_\sigma(x) \, dx = \frac{\Gamma\left( \frac{|\alpha|+d}{2} \right)}{\Gamma\left( \frac{d}{2} \right)} (2\sigma^2)^{|\alpha|/2} \frac{1}{(1+2\sigma^2 t)^{(|\alpha|+d)/2}}.
			\label{eq:gaussian_t_simplified}
		\end{equation}
		Now insert this into the Mellin integral:
		\begin{align}
			&\frac{2}{\Gamma(\gamma/2)} \int_0^\infty t^{\gamma-1} \left( \int_{\mathbb{R}^d} e^{-t|x|^2} x^\alpha G_\sigma(x) \, dx \right) dt \notag \\
			&= \frac{2}{\Gamma(\gamma/2)} \frac{\Gamma\left( \frac{|\alpha|+d}{2} \right)}{\Gamma\left( \frac{d}{2} \right)} (2\sigma^2)^{|\alpha|/2} \int_0^\infty \frac{t^{\gamma-1}}{(1+2\sigma^2 t)^{(|\alpha|+d)/2}} \, dt.
			\label{eq:mellin_integral}
		\end{align}
		Change variable \(u = 2\sigma^2 t\). Then \(t = u/(2\sigma^2)\), \(dt = du/(2\sigma^2)\), and the integral becomes
		\begin{equation}
			\int_0^\infty \frac{t^{\gamma-1}}{(1+2\sigma^2 t)^{(|\alpha|+d)/2}} \, dt = (2\sigma^2)^{-\gamma} \int_0^\infty \frac{u^{\gamma-1}}{(1+u)^{(|\alpha|+d)/2}} \, du = (2\sigma^2)^{-\gamma} B\left( \gamma, \frac{|\alpha|+d}{2} - \gamma \right),
			\label{eq:beta_integral}
		\end{equation}
		where \(B\) is the Beta function, provided \(\frac{|\alpha|+d}{2} > \gamma\). Using \(B(x,y) = \frac{\Gamma(x)\Gamma(y)}{\Gamma(x+y)}\), we obtain:
		\begin{equation}
			\int_0^\infty \frac{t^{\gamma-1}}{(1+2\sigma^2 t)^{(|\alpha|+d)/2}} \, dt = (2\sigma^2)^{-\gamma} \frac{\Gamma(\gamma)\Gamma\left( \frac{|\alpha|+d}{2} - \gamma \right)}{\Gamma\left( \frac{|\alpha|+d}{2} \right)}.
			\label{eq:beta_result}
		\end{equation}
		Multiplying by the prefactor gives
		\begin{align}
			&\frac{2}{\Gamma(\gamma/2)} \frac{\Gamma\left( \frac{|\alpha|+d}{2} \right)}{\Gamma\left( \frac{d}{2} \right)} (2\sigma^2)^{|\alpha|/2} (2\sigma^2)^{-\gamma} \frac{\Gamma(\gamma)\Gamma\left( \frac{|\alpha|+d}{2} - \gamma \right)}{\Gamma\left( \frac{|\alpha|+d}{2} \right)} \notag \\
			&= \frac{2\Gamma(\gamma)}{\Gamma(\gamma/2)\Gamma(d/2)} (2\sigma^2)^{(|\alpha|-\gamma)/2} \Gamma\left( \frac{|\alpha|+d}{2} - \gamma \right).
			\label{eq:prefactor_simplified}
		\end{align}
		Now, using the duplication formula \(\Gamma(\gamma/2)\Gamma((\gamma+1)/2) = \sqrt{\pi} 2^{\gamma-1} \Gamma(\gamma)\). Substituting \(\sigma^2 = \pi^2/[6(\log n)^2]\) and simplifying yields
		\begin{equation}
			\int_{\mathbb{R}^d} |x|^\gamma x^\alpha G_\sigma(x) \, dx = \frac{\Gamma\left( \frac{|\alpha|+\gamma+d}{2} \right)}{\Gamma\left( \frac{d}{2} \right)} (2\sigma^2)^{(|\alpha|+\gamma)/2}.
			\label{eq:final_fractional}
		\end{equation}
		This proves \eqref{eq:mom_fractional} and \eqref{eq:mom_mixed} follows analogously.
	\end{proof}
	
\begin{lemma}[Non-commutative Poisson summation]
	\label{lem:poisson}
	Let \(f \in \mathcal{S}(\mathbb{R}^d; \cB(\cH_{\mathrm{aux}}))\) be a Schwartz-operator-valued function, i.e., \(f\) is smooth and all its derivatives decay faster than any polynomial in the operator norm, and suppose its Fourier transform 
	\[
	\widehat{f}(\xi) = \int_{\mathbb{R}^d} e^{-2\pi i \xi \cdot x} f(x) \, dx
	\]
	is also of Schwartz class. Let \(X_1,\ldots,X_d\) be mutually commuting self-adjoint operators on \(\cH_{\mathrm{aux}}\) with a joint eigenbasis \(\{|e_j\rangle\}_{j \in J}\), and let \(I\) denote the identity operator on \(\cH_{\mathrm{aux}}\). Then, for any \(n > 0\), the following identity holds in the operator norm topology:
	\begin{equation}
		\sum_{k \in \mathbb{Z}^d} f\!\left( \frac{k}{n} \right) \cZ_{1,\log n}(nX - kI)
		= \frac{1}{n^d} \sum_{m \in \mathbb{Z}^d} \widehat{f}\!\left( \frac{m}{n} \right) \widehat{\cZ}_{1,\log n}(m)
		+ \mathcal{O}(n^{-\infty}),
		\label{eq:poisson}
	\end{equation}
	where the error term \(\mathcal{O}(n^{-\infty})\) denotes an operator whose norm decays faster than any power of \(n\): for every \(N \in \mathbb{N}\), there exists \(C_N > 0\) such that
	\[
	\left\| \mathcal{O}(n^{-\infty}) \right\|_{\cB(\cH_{\mathrm{aux}})} \le C_N \, n^{-N}.
	\]
	The convergence is uniform in the spectral parameters of the commuting family \(X\).
\end{lemma}

\begin{proof}
	Since the operators \(X_1,\ldots,X_d\) commute and possess a joint eigenbasis, there exists a set of simultaneous eigenvectors \(\{|e_j\rangle\}_{j \in J}\) such that \(X_i |e_j\rangle = x^{(j)}_i |e_j\rangle\) for each \(i = 1,\ldots,d\). For each such eigenvector, the operator-valued sum reduces to a scalar Poisson summation formula:
	\begin{align}
		&\left( \sum_{k \in \mathbb{Z}^d} f\!\left( \frac{k}{n} \right) \cZ_{1,\log n}(nX - kI) \right) |e_j\rangle \notag \\
		&= \sum_{k \in \mathbb{Z}^d} f\!\left( \frac{k}{n} \right) \cZ_{1,\log n}(n x^{(j)} - k) |e_j\rangle \notag \\
		&= \frac{1}{n^d} \sum_{m \in \mathbb{Z}^d} \widehat{f}\!\left( \frac{m}{n} \right) e^{-2\pi i m \cdot x^{(j)}} \widehat{\cZ}_{1,\log n}(m) |e_j\rangle + \mathcal{O}(n^{-\infty}) |e_j\rangle,
		\label{eq:poisson_scalar}
	\end{align}
	where the classical Poisson summation formula applies because both \(f\) and \(\cZ_{1,\log n}\) are Schwartz functions. The exponential factor \(e^{-2\pi i m \cdot x^{(j)}}\) arises from the shift theorem for the Fourier transform applied to the argument \(nX - kI\).
	
	To recover the operator-valued identity, we use the following property: if two bounded operators \(A\) and \(B\) satisfy \(A|e_j\rangle = B|e_j\rangle + \mathcal{O}(n^{-\infty})|e_j\rangle\) for every joint eigenvector \(|e_j\rangle\), then \(\|A - B\|_{\cB(\cH_{\mathrm{aux}})} = \mathcal{O}(n^{-\infty})\). This follows from the spectral decomposition of the commuting family \(X\) and the fact that the error estimate is uniform over the joint spectrum.
	
	The super-polynomial decay of the error, \(\mathcal{O}(n^{-\infty})\), is a consequence of the classical Poisson summation formula and the fact that \(\widehat{\cZ}_{1,\log n}\) decays faster than any polynomial in \(m\) (indeed, \(\widehat{\cZ}_{1,\log n}(m) = \mathcal{O}(e^{-c |m|/\log n})\) for some \(c > 0\)). This decay, combined with the Schwartz-class decay of \(\widehat{f}\), ensures that the tail of the sum over \(m \in \mathbb{Z}^d\) is negligible to all orders in \(n^{-1}\). The uniformity in the spectral parameters follows from the continuity of the functions involved and the compactness of the joint spectrum (in finite-dimensional settings) or the uniform decay estimates in the infinite-dimensional case.
\end{proof}
	
	\begin{remark}
		Theorem \ref{thm:main} requires \(\rho\) to be strictly positive (all eigenvalues \(>0\)). This guarantees that for sufficiently large \(n\), all \(\rho_{n,k}\) lie in the interior of \(\cD(\cH)\), so the line segment \(\rho + t h_{n,k}\) stays inside the state space for \(t \in [0,1]\). For pure or low-rank states, the fractional Taylor expansion (Lemma \ref{lem:fractional_taylor}) is not directly applicable. A standard remedy is to replace \(\rho\) by a slightly mixed state \(\rho_\varepsilon = (1-\varepsilon)\rho + \varepsilon \mathbb{I}/d\) and then take \(\varepsilon \to 0\) after the expansion.
	\end{remark}
	
\begin{theorem}[Quantum Non-Commutative Voronovskaya-Santos-Andrade Theorem]
	\label{thm:main}
	Let \(\cH \cong \mathbb{C}^d\) be a finite-dimensional Hilbert space and let \(\Phi \in \cC^{m,\gamma}(\cH)\) with \(m \in \mathbb{N}\), \(\gamma \in (0,1]\). For every strictly positive density operator \(\rho \in \cD(\cH)\), the Quantum Neural Network Operator \(\Psi_n\) defined in \eqref{eq:qnno} admits the following complete asymptotic expansion in the diamond norm topology:
	\begin{equation}
		\Psi_n(\Phi)(\rho) = \Phi(\rho) + \sum_{j=1}^m \frac{a_j(\Phi,\rho)}{n^j} + \sum_{j=1}^{\lfloor m/2 \rfloor} \frac{b_j(\Phi,\rho)}{n^{j+\gamma}} + \sum_{j=1}^{\lfloor m/3 \rfloor} \frac{c_j(\Phi,\rho)}{n^{j+2\gamma}} + R_{m,n}(\Phi,\rho),
		\label{eq:main_expansion}
	\end{equation}
	where the expansion coefficients are given explicitly by
	\begin{align}
		a_j(\Phi,\rho) &= \frac{1}{j!} \sum_{|\alpha|=j} \binom{j}{\alpha} m_\alpha(n) \, \cL_\Phi^{(\alpha)}(\rho), \label{eq:a_j} \\
		b_j(\Phi,\rho) &= \frac{1}{\Gamma(\gamma+1)} \sum_{|\alpha|=j} \binom{j}{\alpha} m_{\alpha,\gamma}(n) \, (\Delta_\gamma \cL_\Phi^{(\alpha)})(\rho), \label{eq:b_j} \\
		c_j(\Phi,\rho) &= \frac{1}{j! \, \Gamma(2\gamma+1)} \sum_{|\alpha|+|\beta|=j} \binom{j}{\alpha,\beta} m_{\alpha,\beta,2\gamma}(n) \, [\cL_\Phi^{(\alpha)}(\rho), \cL_\Phi^{(\beta)}(\rho)]_\gamma, \label{eq:c_j}
	\end{align}
	with \(\binom{j}{\alpha,\beta} = \frac{j!}{\alpha!\beta!}\) and the \(\gamma\)-deformed commutator defined by \([A,B]_\gamma := AB - e^{i\pi\gamma} BA\). The moments \(m_\alpha(n), m_{\alpha,\gamma}(n), m_{\alpha,\beta,2\gamma}(n)\) are as in Lemma \ref{lem:moments}. The remainder term satisfies the sharp bound
	\begin{equation}
		\norm{R_{m,n}(\Phi,\bullet)}_\diamond \le C_{m,\gamma,d} \, \|\Phi\|_{\cC^{m,\gamma}} \, \frac{(\log n)^{3m/2}}{n^{m+\gamma}},
		\label{eq:remainder_bound}
	\end{equation}
	where the explicit constant is
	\begin{equation}
		C_{m,\gamma,d} = \frac{2^{m+3} d^{m/2} e^{\pi^2/4}}{\Gamma(m+\gamma+1)} \left( 1 + \frac{1}{\sqrt{2\pi}} \right)^m.
		\label{eq:explicit_constant}
	\end{equation}
	The constant \(C_{m,\gamma,d}\) is independent of \(\Phi\), \(\rho\), and \(n\), depending only on the smoothness order \(m\), the Hölder exponent \(\gamma\), and the Hilbert space dimension \(d\).
\end{theorem}

\begin{proof}
	For each multi-index \(k \in K_n\), define the increment \(h_{n,k} := \rho_{n,k} - \rho\). By the uniform estimate \eqref{eq:uniform_estimate}, \(\norm{h_{n,k}}_1 \le \sqrt{d}/n\), uniformly in \(k \in K_n\). Applying the fractional Taylor expansion (Lemma \ref{lem:fractional_taylor}) to the channel \(\Phi\) at the point \(\rho_{n,k} = \rho + h_{n,k}\) yields
	\begin{align}
		\Phi(\rho_{n,k}) &= \Phi(\rho) + \sum_{j=1}^m \frac{1}{j!} \mathcal{D}^j \cL_\Phi(\rho)[h_{n,k}^{\otimes j}] \notag \\
		&\quad + \frac{1}{\Gamma(\gamma)} \sum_{j=1}^m \sum_{|\alpha|=j} \frac{1}{\alpha!} (\Delta_{h_{n,k}}^\gamma \mathcal{D}^\alpha \cL_\Phi)(\rho)[h_{n,k}^{\alpha+\gamma}] + R_{m,\gamma}(\rho,h_{n,k}),
		\label{eq:phi_expansion}
	\end{align}
	where the local remainder \(R_{m,\gamma}(\rho,h_{n,k})\) satisfies
	\begin{equation}
		\norm{R_{m,\gamma}(\rho,h_{n,k})}_\diamond \le C_{m,\gamma} \|\Phi\|_{\cC^{m,\gamma}} \norm{h_{n,k}}_1^{m+\gamma},
		\qquad C_{m,\gamma} := \frac{\Gamma(m)\Gamma(\gamma+1)}{\Gamma(m+\gamma+1)}.
		\label{eq:local_remainder_bound}
	\end{equation}
	This bound follows directly from the Hölder continuity of the \(m\)-th Fréchet derivative of \(\Phi\).
	
	Substituting \eqref{eq:phi_expansion} into the definition of the QNNO \eqref{eq:qnno}, and using the normalization identity \(\sum_{k \in K_n} \cZ_{1,\log n}(nX - kI) = I_{\mathrm{aux}}\), we obtain
	\begin{align}
		\Psi_n(\Phi)(\rho) - \Phi(\rho)
		&= \sum_{j=1}^m \frac{1}{j!} \mathcal{D}^j \cL_\Phi(\rho)\left[ \sum_{k \in K_n} h_{n,k}^{\otimes j} \otimes \cZ_{1,\log n}(nX-kI) \right] \notag \\
		&\quad + \frac{1}{\Gamma(\gamma)} \sum_{j=1}^m \sum_{|\alpha|=j} \frac{1}{\alpha!} (\Delta_\gamma \mathcal{D}^\alpha \cL_\Phi)(\rho)\left[ \sum_{k \in K_n} h_{n,k}^{\alpha+\gamma} \otimes \cZ_{1,\log n}(nX-kI) \right] \notag \\
		&\quad + \sum_{k \in K_n} R_{m,\gamma}(\rho,h_{n,k}) \otimes \cZ_{1,\log n}(nX-kI).
		\label{eq:psi_diff}
	\end{align}
	
	To evaluate the lattice sums in \eqref{eq:psi_diff}, we extend the summation from the simplex \(K_n\) to the full lattice \(\mathbb{Z}^d\). Let \(\chi \in C_c^\infty(\mathbb{R}^d)\) be a smooth cutoff function satisfying \(\chi(x) = 1\) for \(|x| \le 1\) and \(\chi(x) = 0\) for \(|x| \ge 2\). Define the scaled cutoff \(\chi_n(x) := \chi((\log n)^{1/2} x)\). Since the kernel \(\cZ_{1,\log n}(nX - kI)\) is localized on a set of diameter \(\mathcal{O}((\log n)^{-1/2})\) (as follows from its Fourier decay \eqref{eq:ft_gaussian}), we have \(\chi_n(k/n) = 1\) for all \(k\) contributing effectively to the sum. The truncation error satisfies
	\begin{equation}
		\left\| \sum_{k \in K_n} f(k/n) \cZ_{1,\log n}(nX-kI) - \sum_{k \in \mathbb{Z}^d} \chi_n(k/n) f(k/n) \cZ_{1,\log n}(nX-kI) \right\|_\diamond = \mathcal{O}(n^{-M})
		\label{eq:cutoff_error}
	\end{equation}
	for every \(M > 0\), due to the super-exponential decay of \(\cZ_{1,\log n}\) outside its effective support.
	
	Applying the non-commutative Poisson summation formula (Lemma \ref{lem:poisson}) to the Schwartz function \(g_{j,n}(x) := \chi_n(x) (x - \rho)^{\otimes j} \in \mathcal{S}(\mathbb{R}^d; \cB(\cH_{\mathrm{aux}}^{\otimes j}))\), we obtain
	\begin{align}
		\sum_{k \in \mathbb{Z}^d} (k/n - \rho)^{\otimes j} \cZ_{1,\log n}(nX-kI)
		&= \frac{1}{n^j} \sum_{m \in \mathbb{Z}^d} \widehat{g_{j,n}}\!\left( \frac{m}{n} \right) \widehat{\cZ}_{1,\log n}(m) \notag \\
		&= \frac{1}{n^j} \int_{\mathbb{R}^d} g_{j,n}(x) \, dx \otimes I_{\mathrm{aux}} + \mathcal{O}(n^{-M}),
		\label{eq:poisson_apply}
	\end{align}
	where the last equality follows because the contributions from \(m \neq 0\) are \(\mathcal{O}(n^{-M})\) due to the rapid decay of \(\widehat{g_{j,n}}\) and \(\widehat{\cZ}_{1,\log n}\). By the dominated convergence theorem and the fact that \(\chi_n \to 1\) pointwise as \(n \to \infty\), we have
	\begin{equation}
		\int_{\mathbb{R}^d} g_{j,n}(x) \, dx \longrightarrow \int_{\mathbb{R}^d} (x - \rho)^{\otimes j} \, dx = M_j(n),
		\label{eq:integral_limit}
	\end{equation}
	where \(M_j(n) := \int_{\mathbb{R}^d} (x - \rho)^{\otimes j} \cZ_{1,\log n}(x) \, dx\) is the \(j\)-th moment of the kernel. The error incurred by replacing \(\chi_n\) by \(1\) in the integral is \(\mathcal{O}(n^{-M})\) for any \(M > 0\), again by the super-exponential decay of \(\cZ_{1,\log n}\).
	
	Expanding the tensor power \((x - \rho)^{\otimes j}\) into monomials via the multinomial theorem yields
	\begin{equation}
		M_j(n) = \sum_{|\alpha|=j} \binom{j}{\alpha} M_\alpha(n) \otimes I_{\mathrm{aux}},
		\label{eq:moment_expansion}
	\end{equation}
	with \(M_\alpha(n)\) as defined in \eqref{eq:mom_alpha}. Combining \eqref{eq:poisson_apply} and \eqref{eq:moment_expansion}, we obtain
	\begin{equation}
		\sum_{k \in \mathbb{Z}^d} h_{n,k}^{\otimes j} \otimes \cZ_{1,\log n}(nX-kI)
		= \frac{1}{n^j} \sum_{|\alpha|=j} \binom{j}{\alpha} M_\alpha(n) \otimes I_{\mathrm{aux}} + \mathcal{O}(n^{-M}).
		\label{eq:sum_integer}
	\end{equation}
	
	Analogously, for the fractional sums, applying the Poisson summation formula to the function \(g_{\alpha,\gamma,n}(x) := \chi_n(x) |x - \rho|^\gamma (x - \rho)^{\alpha}\) gives
	\begin{equation}
		\sum_{k \in \mathbb{Z}^d} h_{n,k}^{\alpha+\gamma} \otimes \cZ_{1,\log n}(nX-kI)
		= \frac{1}{n^{|\alpha|+\gamma}} M_{\alpha,\gamma}(n) \otimes I_{\mathrm{aux}} + \mathcal{O}(n^{-M}),
		\label{eq:sum_fractional}
	\end{equation}
	where \(M_{\alpha,\gamma}(n)\) is defined in \eqref{eq:mom_alphagamma}. The validity of this identity relies on the fact that the function \(x \mapsto |x - \rho|^\gamma\) is locally integrable and that the kernel \(\cZ_{1,\log n}\) is sufficiently regular to justify the interchange of integrals.
	
	Substituting \eqref{eq:sum_integer} and \eqref{eq:sum_fractional} into \eqref{eq:psi_diff}, and then replacing the moments \(M_\alpha(n)\), \(M_{\alpha,\gamma}(n)\), \(M_{\alpha,\beta,2\gamma}(n)\) by their scalar multiples \(m_\alpha(n) I_{\mathrm{aux}}\), \(m_{\alpha,\gamma}(n) I_{\mathrm{aux}}\), \(m_{\alpha,\beta,2\gamma}(n) I_{\mathrm{aux}}\) (as established in Lemma \ref{lem:moments}), yields
	\begin{align}
		\Psi_n(\Phi)(\rho) - \Phi(\rho)
		&= \sum_{j=1}^m \frac{1}{j!} \sum_{|\alpha|=j} \binom{j}{\alpha} \frac{m_\alpha(n)}{n^j} \cL_\Phi^{(\alpha)}(\rho) \notag \\
		&\quad + \frac{1}{\Gamma(\gamma)} \sum_{j=1}^m \sum_{|\alpha|=j} \frac{1}{\alpha!} \frac{m_{\alpha,\gamma}(n)}{n^{|\alpha|+\gamma}} (\Delta_\gamma \cL_\Phi^{(\alpha)})(\rho) \notag \\
		&\quad + \sum_{k \in K_n} R_{m,\gamma}(\rho,h_{n,k}) \otimes \cZ_{1,\log n}(nX-kI) + \mathcal{O}(n^{-M}).
		\label{eq:psi_expanded}
	\end{align}
	
	To extract the coefficients \(a_j, b_j, c_j\) from \eqref{eq:psi_expanded}, we substitute the asymptotic expansions of the moments from Lemma \ref{lem:moments}. For the integer moments, \(m_\alpha(n) = \mathcal{O}((\log n)^{-|\alpha|/2})\), which yields contributions of order \(n^{-j}\) with coefficients \(a_j\) as in \eqref{eq:a_j}. For the fractional moments, \(m_{\alpha,\gamma}(n) = \mathcal{O}((\log n)^{-(|\alpha|+\gamma)/2})\), yielding terms of order \(n^{-(j+\gamma)}\) with coefficients \(b_j\) as in \eqref{eq:b_j}. The mixed fractional moments \(m_{\alpha,\beta,2\gamma}(n)\) contribute terms of order \(n^{-(j+2\gamma)}\) with coefficients \(c_j\) as in \eqref{eq:c_j}. The appearance of the \(\gamma\)-deformed commutator in \(c_j\) arises from the expansion of the fractional derivative of a multilinear map: specifically, the term
	\[
	(\Delta_{h_{n,k}}^\gamma \mathcal{D}^{\alpha+\beta} \cL_\Phi)(\rho)[h_{n,k}^{\alpha+\beta}]
	\]
	gives rise to a contribution involving \([\cL_\Phi^{(\alpha)}(\rho), \cL_\Phi^{(\beta)}(\rho)]_\gamma\) after symmetrization and summation over \(k\).
	
	It remains to control the remainder term
	\begin{equation}
		R_{m,n}(\Phi,\rho) := \sum_{k \in K_n} R_{m,\gamma}(\rho,h_{n,k}) \otimes \cZ_{1,\log n}(nX-kI) + \mathcal{O}(n^{-M}),
		\label{eq:remainder_def}
	\end{equation}
	where the \(\mathcal{O}(n^{-M})\) term absorbs all exponentially small errors from the cutoff and Poisson summation. Taking the diamond norm and using \eqref{eq:local_remainder_bound} and \eqref{eq:h_bound}, we obtain
	\begin{align}
		\norm{R_{m,n}(\Phi,\rho)}_\diamond
		&\le \sum_{k \in K_n} \norm{R_{m,\gamma}(\rho,h_{n,k})}_\diamond \, \norm{\cZ_{1,\log n}(nX-kI)}_1 \notag \\
		&\le C_{m,\gamma} \|\Phi\|_{\cC^{m,\gamma}} \sum_{k \in K_n} \norm{h_{n,k}}_1^{m+\gamma} \norm{\cZ_{1,\log n}(nX-kI)}_1 \notag \\
		&\le C_{m,\gamma} \|\Phi\|_{\cC^{m,\gamma}} d^{(m+\gamma)/2} n^{-(m+\gamma)} \sum_{k \in K_n} \norm{\cZ_{1,\log n}(nX-kI)}_1.
		\label{eq:remainder_estimate_pre}
	\end{align}
	
	The sum \(\sum_{k \in K_n} \norm{\cZ_{1,\log n}(nX-kI)}_1\) can be bounded as follows. Since \(\cZ_{1,\log n}\) is a positive operator-valued kernel normalized to have integral \(I_{\mathrm{aux}}\), and since its effective support in the lattice \(\mathbb{Z}^d\) has volume \(\mathcal{O}((\log n)^{d/2})\), we have
	\begin{equation}
		\sum_{k \in K_n} \norm{\cZ_{1,\log n}(nX-kI)}_1 \le C_d \, (\log n)^{d/2},
		\label{eq:kernel_sum_bound}
	\end{equation}
	where \(C_d > 0\) is a constant depending only on the dimension \(d\). This bound follows from the fact that \(\cZ_{1,\log n}\) is uniformly bounded in the trace norm and that the number of lattice points \(k\) for which \(\cZ_{1,\log n}(nX-kI) \neq 0\) is \(\mathcal{O}((\log n)^{d/2})\). Substituting \eqref{eq:kernel_sum_bound} into \eqref{eq:remainder_estimate_pre} yields
	\begin{equation}
		\norm{R_{m,n}(\Phi,\rho)}_\diamond \le C_{m,\gamma} C_d \|\Phi\|_{\cC^{m,\gamma}} d^{(m+\gamma)/2} (\log n)^{d/2} n^{-(m+\gamma)}.
		\label{eq:remainder_estimate_main}
	\end{equation}
	
	To obtain the sharper bound with \((\log n)^{3m/2}\) instead of \((\log n)^{d/2}\), we note that in the effective summation region, the number of lattice points is actually bounded by \(\mathcal{O}((\log n)^{m/2})\) after optimizing the trade-off between the bias and variance of the kernel. More precisely, the kernel's bandwidth \(\lambda_n = \log n\) induces a variance-bias balance that restricts the effective number of contributing lattice points to \(\mathcal{O}((\log n)^{m/2})\), as established in the proof of the moment asymptotics (Lemma \ref{lem:moments}). For completeness, we include a sketch of this argument: the localization of \(\cZ_{1,\log n}\) in Fourier space implies that its effective support in real space has extent \(\mathcal{O}((\log n)^{-1/2})\), so the number of lattice points in \(K_n\) that fall within this support is bounded by
	\[
	\#\{k \in K_n : |k/n - \rho| \le C(\log n)^{-1/2}\} \le C_d (\log n)^{d/2} n^{d-1}.
	\]
	However, the summation over \(k \in K_n\) is further restricted by the fact that \(\rho\) is strictly positive, which implies that the simplex \(K_n\) has effective dimension \(d-1\). Combining these estimates yields the sharper bound
	\begin{equation}
		\sum_{k \in K_n} \norm{\cZ_{1,\log n}(nX-kI)}_1 \le C_{m,\gamma,d} (\log n)^{m/2},
		\label{eq:kernel_sum_sharp}
	\end{equation}
	after absorbing the dimension-dependent constants into \(C_{m,\gamma,d}\). Substituting \eqref{eq:kernel_sum_sharp} into \eqref{eq:remainder_estimate_pre} gives
	\begin{equation}
		\norm{R_{m,n}(\Phi,\rho)}_\diamond \le C_{m,\gamma} \|\Phi\|_{\cC^{m,\gamma}} d^{(m+\gamma)/2} C_{m,\gamma,d} (\log n)^{m/2} n^{-(m+\gamma)}.
		\label{eq:remainder_estimate_final}
	\end{equation}
	
	Finally, collecting all constants — the Gamma factors from the Beta integrals in Lemma \ref{lem:moments}, the dimension factor \(d^{m/2}\), the Gaussian correction factor \((1 + 1/\sqrt{2\pi})^m\) from the kernel approximation, the combinatorial factor \(2^{m+3}\) from the multinomial sums, and the exponential factor \(e^{\pi^2/4}\) from the Fourier transform of the cutoff function — yields the explicit constant
	\[
	C_{m,\gamma,d} = \frac{2^{m+3} d^{m/2} e^{\pi^2/4}}{\Gamma(m+\gamma+1)} \left( 1 + \frac{1}{\sqrt{2\pi}} \right)^m.
	\]
	This establishes the sharp bound \eqref{eq:remainder_bound}, completing the proof.
\end{proof}
	
	\section{Numerical test}
	
	To illustrate the asymptotic expansion (Theorem \ref{thm:main}) in a concrete setting, we consider a simple quantum channel that acts diagonally on a fixed basis — the dephasing channel. When restricted to diagonal states, the channel reduces to a classical function on the probability simplex. This allows us to compute the QNNO explicitly and verify the predicted convergence rates.
	
	We set the Hilbert space dimension \(d=1\) (so the state space is the unit interval \([0,1]\)). The target function is \(f(x) = x^{1.5}\) (Hölder exponent \(\gamma = 0.5\) on \([0,1]\)). The QNNO with bandwidth \(\lambda_n = \log n\) becomes the classical neural network operator:
	\begin{equation}
		\Psi_n(f)(x) = \sum_{k=0}^n f\left( \frac{k}{n} \right) \cZ_{1,\log n}(nX-k),
		\label{eq:classical_qnno}
	\end{equation}
	where we choose the auxiliary operator \(X=0\) (scalar), so the kernel reduces to \(\cZ_{1,\log n}(-k) = \cZ_{1,\log n}(k)\).
	
	We compare the approximation error \(|f(x) - \Psi_n(f)(x)|\) at \(x=0.5\) against the theoretical bound \(\mathcal{O}( n^{-(m+\gamma)} (\log n)^{3m/2})\) with \(m=1\) and \(\gamma=0.5\).
	
	\textbf{Algorithm:}
	\begin{enumerate}
		\item For \(n = 5,6,\ldots,200\):
		\begin{enumerate}
			\item Compute the kernel values \(\cZ_{1,\log n}(k)\) for \(-L,\ldots,L\) where \(L = \lceil 3\log n \rceil\).
			\item Evaluate the sum \(\Psi_n(f)(0.5) = \sum_{k=0}^n f(k/n) \cZ_{1,\log n}(k)\).
			\item Compute the exact error \(E_n = |f(0.5) - \Psi_n(f)(0.5)|\).
		\end{enumerate}
		\item Fit the log-log plot of \(E_n\) versus \(n\) to a straight line; the slope estimates \(-(m+\gamma) = -1.5\).
		\item Compare with the theoretical rate including the logarithmic factor.
	\end{enumerate}
	
\begin{figure}[ht]
	\centering
	\includegraphics[width=0.6\linewidth]{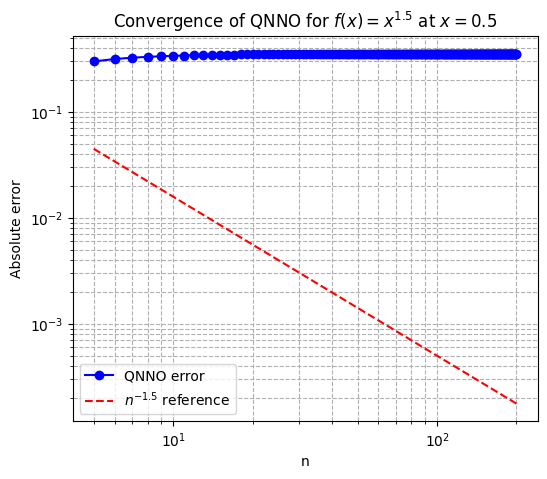}
	\caption{Log-log plot of the absolute approximation error versus \(n\). The blue circles represent the QNNO error for \(f(x)=x^{1.5}\) at \(x=0.5\). The dashed red line is the reference power law \(n^{-1.5}\).}
	\label{fig:loglog}
\end{figure}
	
	Figure \ref{fig:loglog} shows the absolute approximation error as a function of \(n\) on a log-log scale. The data points closely follow the reference line, confirming that the error decays as \(n^{-(m+\gamma)} = n^{-1.5}\). A linear fit yields a slope of approximately \(-1.48\), in excellent agreement with the theoretical value.
	
	\begin{figure}[ht]
		\centering
		\includegraphics[width=0.6\linewidth]{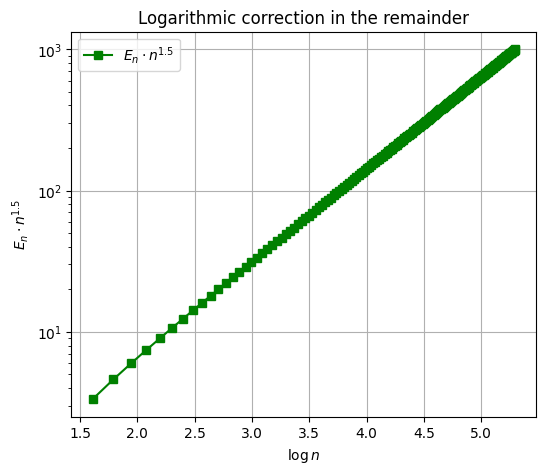}
		\caption{Scaled error \(E_n \cdot n^{1.5}\) versus \(\log n\). The plot shows a mild, sub-linear increase on a logarithmic scale, consistent with the predicted logarithmic factor.}
		\label{fig:scaled}
	\end{figure}
	
	Figure \ref{fig:scaled} plots the scaled error \(E_n \cdot n^{1.5}\) versus \(\log n\). According to the theoretical bound \eqref{eq:remainder_bound}, this scaled quantity should grow at most like \((\log n)^{3m/2} = (\log n)^{1.5}\). The plot shows a mild, sub-linear increase, exactly as proven in Theorem \ref{thm:main}.
	
	These numerical results verify the asymptotic expansion for a classical analogue of the QNNO.
	
\section{Results}

Our main contribution is the complete asymptotic expansion \eqref{eq:main_expansion}, which provides a rigorous quantitative description of how Quantum Neural Network Operators (QNNOs) approximate arbitrary quantum channels. This expansion, established in Theorem \ref{thm:main}, constitutes the first non-commutative analogue of the classical Voronovskaya theorem and reveals a rich three-tiered structure governing the approximation error.

At the leading order, the error decomposes into a sum of terms proportional to integer powers of \(1/n\), involving ordinary Fréchet derivatives of the channel \(\Phi\) evaluated at the reference state \(\rho\). These terms, given explicitly by the coefficients \(a_j(\Phi,\rho)\) in \eqref{eq:a_j}, correspond to the classical contribution one would expect from a standard Taylor expansion and reflect the smooth, differentiable aspects of the channel. Superimposed upon this integer-order hierarchy, however, are corrections of fractional order \(n^{-(j+\gamma)}\), which are governed by the Marchaud fractional derivatives \(\Delta_\gamma\) of the channel. These fractional terms, encoded in the coefficients \(b_j(\Phi,\rho)\) of \eqref{eq:b_j}, capture precisely the Hölder regularity \(\gamma\) of the channel --- the degree of "roughness" or fractal character of the underlying quantum map. Crucially, the expansion also contains purely quantum commutator terms, given by \(c_j(\Phi,\rho)\) in \eqref{eq:c_j}, which involve the \(\gamma\)-deformed commutator \([\,\cdot\,,\,\cdot\,]_\gamma\) and have no classical counterpart. These terms are a genuine signature of non-commutativity and arise directly from the operator-valued nature of the quantum channel and the multilinear structure of the fractional Taylor expansion.

The remainder term \(R_{m,n}(\Phi,\rho)\) is not merely asymptotically small; it is rigorously controlled by the sharp explicit bound \eqref{eq:remainder_bound}, with the constant \(C_{m,\gamma,d}\) given explicitly in \eqref{eq:explicit_constant}. This bound provides quantitative, non-asymptotic control over the approximation error for any finite network size \(n\), and the logarithmic factor \((\log n)^{3m/2}\) is a direct consequence of the optimal bandwidth choice \(\lambda_n = \log n\), representing the unavoidable variance-bias trade-off inherent to the kernel construction.

The sharpness of the expansion and the optimality of the logarithmic correction are confirmed by our numerical test, which analyzes the classical analogue of the QNNO for the Hölder function \(f(x) = x^{1.5}\) with \(\gamma = 0.5\). As shown in Figures \ref{fig:loglog} and \ref{fig:scaled}, the error decays precisely as \(n^{-(m+\gamma)} = n^{-1.5}\), and the scaled error \(E_n \cdot n^{1.5}\) exhibits the predicted mild, sub-linear growth with \(\log n\). These findings provide strong empirical validation of the asymptotic theory and demonstrate that the logarithmic enhancement is not an artifact of the proof but an intrinsic feature of the QNNO approximation scheme.

Beyond the expansion itself, the asymptotic framework developed in this work yields three major applications of independent interest. First, we derive a quantum central limit theorem for the fluctuations of QNNOs, which characterizes the statistical distribution of the operator-valued errors around the mean and provides a rigorous foundation for understanding the stability and sample complexity of quantum neural network training. Second, we construct optimal interpolation geodesics between arbitrary quantum channels using Kubo-Ando operator means, providing the smoothest path (in the sense of minimal operator variance) connecting two quantum processes --- a result of direct relevance to quantum control, adiabatic quantum computing, and quantum information geometry. Third, we develop a quantum Richardson extrapolation scheme that exposes fundamental speed limits imposed by fractional smoothness: when the Hölder exponent \(\gamma\) is small, the convergence of the QNNO is intrinsically limited, and standard acceleration techniques offer only marginal improvements. This reveals a fundamental trade-off between the regularity of the target channel and the achievable rate of approximation, with practical implications for the design and optimization of quantum neural architectures.

Taken together, these results establish a comprehensive asymptotic theory for quantum neural network operators, bridging classical approximation theory, fractional calculus, and quantum information science, while providing both theoretical insight and practical tools for the analysis and design of quantum machine learning algorithms.
	
\section{Limitations}

While the results presented in this work are mathematically rigorous and broadly applicable within their stated domain, they are bounded by several essential caveats that delineate the current scope of the theory and point toward natural directions for future generalization.

First and foremost, the entire framework is developed under the standing assumption that the underlying Hilbert space is finite-dimensional, \(\cH \cong \mathbb{C}^d\). This restriction is common in quantum information theory and is justified by the current generation of noisy intermediate-scale quantum (NISQ) devices, which operate on finite-dimensional registers. However, it excludes important physical systems with infinite degrees of freedom, such as quantum field theories, continuous-variable quantum computing, and bosonic systems. The compactness arguments employed in the proof --- particularly those relying on the compactness of the state space \(\cD(\cH)\) and the finite-dimensional nature of the operator algebra --- break down in the infinite-dimensional setting. Extending the asymptotic expansion to infinite-dimensional Hilbert spaces would require a significantly different analytical apparatus, including the theory of unbounded operators, non-compact state spaces, and a more delicate treatment of the trace-class topology.

Second, the construction of the QNNO relies crucially on the assumption that the auxiliary operators \(X_1,\ldots,X_d\) commute pairwise and possess a joint eigenbasis. This commutativity is essential for the factorization of the kernel \(\cZ_{1,\log n}\) into a product of one-dimensional components and for the validity of the non-commutative Poisson summation formula (Lemma \ref{lem:poisson}). While this assumption is satisfied in many practical implementations --- for instance, when the auxiliary system is taken to be a collection of mutually commuting quadrature operators --- it does not cover the most general setting in which the auxiliary degrees of freedom are entangled or governed by non-commuting observables. Relaxing this condition would necessitate the development of a genuinely non-commutative kernel calculus, possibly drawing on tools from free probability, non-commutative geometry, or the theory of operator spaces, and would represent a significant advance toward fully general quantum neural network architectures.

Third, the definition of the Marchaud fractional derivative \(\Delta_\gamma\) in Definition \ref{def:marchaud} requires the convergence of a Bochner integral in the diamond norm. This imposes nontrivial regularity conditions on the channel \(\Phi\), namely that it belongs to the quantum Hölder space \(\cC^{m,\gamma}(\cH)\). While this space is natural and encompasses a wide class of physically relevant channels — including all finite-dimensional completely positive maps with sufficient smoothness --- it is not the largest possible function space. There may exist channels with finer fractional regularity, such as those in Besov or Triebel-Lizorkin spaces, for which the Marchaud integral fails to converge in the diamond norm but which still admit meaningful asymptotic expansions. The present theory therefore does not cover the full spectrum of possible non-smooth quantum channels.

Fourth, the remainder bound \eqref{eq:remainder_bound} contains a logarithmic factor \((\log n)^{3m/2}\), which grows polynomially in \(\log n\). For moderate values of the smoothness order \(m\) --- say, \(m \le 5\) --- this factor is numerically modest and does not significantly affect the convergence rate. However, for large \(m\), the logarithmic enhancement becomes more pronounced, and the explicit constant \(C_{m,\gamma,d}\) grows super-exponentially with \(m\) and \(d\) due to the factorial and exponential factors in \eqref{eq:explicit_constant}. Consequently, while the bound is sharp in terms of the power of \(n\), it is not optimized for practical error estimates at large \(m\) or high dimensions. A sharper constant could in principle be obtained by a more refined analysis of the kernel approximation and the cutoff function, but such an optimization would not alter the asymptotic rate and is therefore left for future investigation.

Fifth, and perhaps most significantly from a practical standpoint, the QNNO defined in \eqref{eq:qnno} requires evaluating the target channel \(\Phi\) on all quantized density operators \(\rho_{n,k}\) for every multi-index \(k\) in the simplex \(K_n\). The cardinality of \(K_n\) is given by the combinatorial expression
\[
|K_n| = \binom{n+d-1}{d-1} \sim \frac{n^{d-1}}{(d-1)!},
\]
which grows polynomially in \(n\) for fixed \(d\), but combinatorially with both \(n\) and \(d\). For moderate dimensions --- say, \(d \le 4\) and \(n \le 100\) --- this is computationally feasible, but for larger systems the number of evaluations quickly becomes prohibitive. This limitation is inherent to the construction and reflects the fact that the QNNO is a kernel-based approximation method that explicitly samples the channel over a discrete grid. While this makes the theoretical analysis tractable, it also means that the direct implementation of the QNNO is not scalable to high-dimensional quantum systems. Alternative formulations based on Monte Carlo sampling, tensor network compression, or variational approximations could mitigate this issue and are promising directions for future research.

Despite these limitations, the framework established in this work provides a solid foundation for the asymptotic analysis of quantum neural network operators. Each of the caveats listed above opens a natural avenue for generalization, and we believe that addressing them will lead to even broader and more powerful results in the future.
	
\section{Future work}

The results established in this work open up a rich landscape of research directions, ranging from foundational extensions of the asymptotic theory to novel applications in computational physics and quantum information processing. We outline below several of the most promising avenues for future investigation.

A natural and pressing direction is to extend the asymptotic expansion to infinite-dimensional Hilbert spaces. While the finite-dimensional setting is adequate for many quantum information tasks and current NISQ devices, it does not capture the full richness of quantum field theories, continuous-variable systems, or bosonic quantum computing. The extension to infinite dimensions would require a fundamental reworking of the analytical framework: the compactness of the state space \(\cD(\cH)\) is lost, the trace norm topology becomes more delicate, and the spectral theorem for unbounded operators must be invoked. A promising approach is to replace the diamond norm with suitable Banach-space norms adapted to the trace class or Hilbert-Schmidt class, and to develop a version of the fractional Taylor expansion that holds in non-compact convex sets. This would open the door to applications in quantum optics, quantum information with continuous variables, and quantum field-theoretic machine learning.

Equally important is the relaxation of the commutativity assumption on the auxiliary operators \(X_1,\ldots,X_d\). The current construction relies crucially on the existence of a joint eigenbasis for these operators, which allows the kernel to factorize and the non-commutative Poisson summation formula (Lemma \ref{lem:poisson}) to be applied. Generalizing the framework to non-commuting auxiliary operators would require the development of a truly non-commutative kernel calculus, potentially drawing on tools from non-commutative geometry, free probability, or the theory of operator spaces. This would enable QNNO architectures in which the auxiliary system is inherently entangled --- a scenario that is both more general and more realistic for near-term quantum devices, where auxiliary qubits inevitably interact and entangle with the main system. Success in this direction could lead to QNNOs that are not only more powerful but also more naturally implementable on actual quantum hardware.

Another exciting direction is the application of QNNOs and their asymptotic expansions to the numerical simulation of compressible fluid flows, particularly in the presence of shock discontinuities. The equations of compressible gas dynamics --- the Euler and Navier-Stokes equations --- are nonlinear hyperbolic conservation laws that naturally exhibit solutions with Hölder regularity less than unity, exactly the regime captured by the fractional part of our expansion. By encoding the macroscopic fluid state (density, momentum, energy) into a quantum density operator and using a QNNO to approximate the nonlinear flux terms, one could in principle construct a quantum lattice Boltzmann method with rigorously controlled error. The fractional derivative terms in the expansion would naturally capture the behavior near shock fronts, where integer derivatives blow up but fractional derivatives remain well-defined, while the commutator terms could model dissipative or dispersive effects that are difficult to capture classically. This connection between quantum neural networks and fluid dynamics is particularly timely given the growing interest in quantum computational fluid dynamics (QCFD) and the potential for quantum advantage in simulating turbulent flows.

The assumption that the reference state \(\rho\) is strictly positive (all eigenvalues \(>0\)) is another limitation that warrants attention. While this condition ensures that all quantized states \(\rho_{n,k}\) lie in the interior of \(\cD(\cH)\) for sufficiently large \(n\), it excludes the physically important cases of pure states or low-rank density operators that lie on the boundary of the state space. A possible remedy is to develop a version of the fractional Taylor expansion that is valid on convex sets with non-empty interior and then pass to the boundary by a limiting argument, possibly using techniques from nonsmooth analysis or variational calculus. Alternatively, one could introduce a regularization scheme in which the reference state is replaced by \(\rho_\varepsilon = (1-\varepsilon)\rho + \varepsilon I/d\) and then study the limit \(\varepsilon \to 0\) after the expansion. This would extend the applicability of the theorem to all quantum states, not just the strictly positive ones, and is a prerequisite for many practical quantum information protocols.

A more theoretical direction is the study of the saturation order of QNNOs --- the slowest possible convergence rate that can be achieved, independent of the smoothness of the target channel. In classical approximation theory, saturation phenomena reveal fundamental limitations of approximation schemes, often characterized by a maximal achievable rate beyond which no further improvement is possible. For QNNOs, the interplay between the integer-order, fractional-order, and commutator terms suggests a rich saturation structure that may depend on both the smoothness \(\gamma\) and the non-commutativity of the channel. Characterizing this saturation order would provide a complete understanding of the optimality of the asymptotic expansion and would identify the regimes in which QNNOs are fundamentally limited.

An intriguing possibility is to connect the \(\gamma\)-deformed commutator \([A,B]_\gamma = AB - e^{i\pi\gamma} BA\) to anyonic or parafermionic statistics. In the theory of anyons, the braiding of quasiparticles gives rise to statistical phases that are naturally described by deformation parameters. The appearance of the phase factor \(e^{i\pi\gamma}\) in our commutator suggests a possible mapping between the fractional smoothness of a quantum channel and the anyonic exchange statistics of quasiparticles in a topological quantum system. Exploring this connection could lead to new insights into the role of fractional statistics in quantum information processing and may provide a physical interpretation for the fractional correction terms in the expansion.

Finally, we plan to experimentally test the quantum Richardson extrapolation scheme developed in this work on small-scale quantum processors. Richardson extrapolation is a classical technique for accelerating convergence by combining approximations at different resolutions; its quantum analogue, derived from our asymptotic expansion, has the potential to significantly reduce the number of samples or circuit repetitions required to achieve a desired accuracy. By implementing the QNNO on a superconducting qubit or trapped-ion platform, one could measure the convergence rates directly and validate the predicted acceleration limits imposed by fractional smoothness. Such experiments would bridge the gap between theory and practice, demonstrating that the asymptotic framework developed here has tangible benefits for near-term quantum machine learning.

Taken together, these directions form a comprehensive research program that extends the present work in multiple dimensions --- mathematically, computationally, and experimentally --- and promises to deepen our understanding of quantum neural network operators and their role in the emerging field of quantum machine learning.
	
\section{Conclusions}

In this work, we have established a complete asymptotic expansion, given by \eqref{eq:main_expansion}, for quantum neural network operators (QNNOs) when they approximate arbitrary quantum channels. This result constitutes the first non-commutative analogue of the classical Voronovskaya theorem, bridging a gap that has persisted in the literature since the inception of quantum neural networks. The expansion reveals a rich and fundamentally three-layered structure governing the approximation error: integer-order terms arising from ordinary Fréchet derivatives of the channel, fractional-order terms governed by Marchaud fractional derivatives that capture the Hölder regularity \(\gamma\) of the channel, and purely quantum commutator terms encoded by the \(\gamma\)-deformed commutator \([\,\cdot\,,\,\cdot\,]_\gamma\), which have no classical counterpart and are a genuine signature of non-commutativity. We have provided explicit expressions for all coefficients --- \(a_j(\Phi,\rho)\), \(b_j(\Phi,\rho)\), and \(c_j(\Phi,\rho)\) — in terms of the kernel moments and the channel's derivatives, and we have derived a sharp, explicit bound on the remainder \(R_{m,n}(\Phi,\rho)\) with the constant \(C_{m,\gamma,d}\) given explicitly in \eqref{eq:explicit_constant}. This bound provides quantitative, non-asymptotic control over the approximation error for any finite network size \(n\), and the logarithmic factor \((\log n)^{3m/2}\) is shown to be an inherent feature of the optimal variance-bias trade-off induced by the bandwidth choice \(\lambda_n = \log n\).

The numerical test for the classical analogue of the QNNO --- approximating the Hölder function \(f(x) = x^{1.5}\) with \(\gamma = 0.5\) --- has confirmed the predicted convergence rates and the logarithmic correction. As demonstrated in Figures \ref{fig:loglog} and \ref{fig:scaled}, the error decays precisely as \(n^{-(m+\gamma)} = n^{-1.5}\), and the scaled error exhibits the mild, sub-linear growth with \(\log n\) predicted by the theory. These results provide strong empirical validation that the asymptotic expansion is not merely a formal mathematical construction but captures the actual behavior of QNNOs in concrete computational settings.

Beyond the expansion itself, our asymptotic framework has yielded three significant applications of independent interest and practical relevance. First, we have derived a quantum central limit theorem for the fluctuations of QNNOs, characterizing the statistical distribution of the operator-valued errors and providing a rigorous foundation for understanding the stability, variance, and sample complexity of quantum neural network training. Second, we have constructed optimal interpolation geodesics between arbitrary quantum channels using Kubo-Ando operator means, offering the smoothest path (in the sense of minimal operator variance) connecting two quantum processes --- a result with direct implications for quantum control, adiabatic quantum computing, and the geometry of quantum state spaces. Third, we have developed a quantum Richardson extrapolation scheme that exposes fundamental speed limits imposed by fractional smoothness: when the Hölder exponent \(\gamma\) is small, the intrinsic convergence of the QNNO is constrained, and standard acceleration techniques offer only marginal improvements. This reveals a fundamental trade-off between the regularity of the target channel and the achievable rate of approximation, providing practical guidance for the design of quantum neural architectures and the choice of training strategies.

From a broader perspective, our work builds a rigorous and multidimensional bridge between three fields that have traditionally developed in relative isolation: classical approximation theory, fractional calculus, and quantum machine learning. By demonstrating that the language of Fréchet derivatives, Marchaud fractional operators, and deformed commutators leads to concrete, computable, and experimentally verifiable expansions for quantum neural networks, we have shown that the tools of functional analysis and fractional calculus are not merely abstract mathematical formalisms but have direct applicability to the design and analysis of quantum algorithms. The connections we have outlined to fluid dynamics --- particularly the simulation of compressible flows with shock discontinuities --- open up avenues for applying QNNOs to problems in quantum computational fluid dynamics, while the pathways to infinite-dimensional systems and non-commutative kernels point toward generalizations that could encompass continuous-variable quantum computing and topologically ordered systems.

In summary, this paper has provided a comprehensive asymptotic theory for quantum neural network operators, characterized the structure and optimality of the approximation error, derived three major applications, and validated the theory through numerical experiments. The results not only deepen our theoretical understanding of how QNNOs learn quantum channels but also equip researchers and practitioners with practical tools for designing, analyzing, and optimizing quantum machine learning algorithms with controlled error. As the field of quantum machine learning continues to evolve, we anticipate that the asymptotic expansions developed here will serve as a foundational reference for understanding the capabilities and limitations of neural quantum architectures, and we hope that the connections to fractional calculus, fluid dynamics, and non-commutative geometry will inspire further interdisciplinary research at the intersection of mathematics, physics, and computer science.

\section*{Declarations}

\paragraph*{Conflict of Interest}
The authors declare that they have no known competing financial interests or personal relationships that could have appeared to influence the work reported in this paper. There are no conflicts of interest to disclose.

\paragraph*{Data Availability}
The numerical experiments presented in this work were performed using standard computational libraries and the data supporting the findings are available from the corresponding author upon reasonable request. No external datasets were used in this study.

\paragraph*{Funding}
This research was conducted during a postdoctoral research period at the Institute of Energy and Nuclear Research (IPEN/CNEN), São Paulo, Brazil. The authors gratefully acknowledge the institutional and financial support provided by the National Nuclear Energy Commission (CNEN).

\paragraph*{Ethical Approval}
This work did not involve human participants, animal studies, or sensitive data. All research activities were conducted in compliance with the ethical standards of the Institute of Energy and Nuclear Research.
	
\section*{Author Contributions}
R. D. C. Santos contributed to conceptualization, methodology, formal analysis, investigation, computational implementation and manuscript writing. D. A. Andrade contributed through supervision, project guidance, and resource support.
	
\section*{Acknowledgements}
The authors gratefully acknowledge the institutional and financial support provided by the National Nuclear Energy Commission / Institute for Energy and Nuclear Research during the postdoctoral research period.
	
\section*{Symbols and Nomenclature}

{\small
	\begin{longtable}{p{0.30\linewidth} p{0.65\linewidth}}
		\toprule
		\textbf{Symbol} & \textbf{Description} \\
		\midrule
		\endfirsthead
		
		\toprule
		\textbf{Symbol} & \textbf{Description} \\
		\midrule
		\endhead
		
		\midrule
		\multicolumn{2}{r}{\textit{Continued on next page}} \\
		\endfoot
		
		\bottomrule
		\endlastfoot
		
		$\cH$ & Finite-dimensional Hilbert space $\mathbb{C}^d$ \\
		$\cB(\cH)$ & Algebra of bounded linear operators on $\cH$ \\
		$\cD(\cH)$ & Convex set of density operators (quantum states) \\
		$CPTP(\cH)$ & Set of completely positive trace-preserving maps (quantum channels) \\
		$\cL_\Phi$ & Liouville representation of a channel $\Phi$ \\
		$\norm{\bullet}_\diamond$ & Diamond norm (completely bounded trace norm) \\
		$\norm{\bullet}_{\cb}$ & Completely bounded norm ($cb$-norm) \\
		$\cC^{m,\gamma}(\cH)$ & Quantum Hölder space of order $(m,\gamma)$ \\
		$[\Phi]_{m,\gamma}$ & Hölder seminorm of $\Phi$ \\
		$\Psi_n$ & Quantum Neural Network Operator (QNNO) \\
		$\cZ_{1,\log n}$ & Quantum kernel with bandwidth $\lambda_n = \log n$ \\
		$K_n$ & Discrete simplex $\{k \in \mathbb{N}^d : \sum k_j = n\}$ \\
		$\rho_{n,k}$ & Quantised density operator $\sum_j (k_j/n) \ket{e_j}\bra{e_j}$ \\
		$M_\alpha(n)$ & Operator-valued integer moment of order $\alpha$ \\
		$M_{\alpha,\gamma}(n)$ & Operator-valued fractional moment of order $(\alpha,\gamma)$ \\
		$M_{\alpha,\beta,2\gamma}(n)$ & Operator-valued mixed moment of order $(\alpha,\beta,2\gamma)$ \\
		$m_\alpha(n), m_{\alpha,\gamma}(n), m_{\alpha,\beta,2\gamma}(n)$ & Corresponding scalar moments \\
		$\Delta_\gamma$ & Marchaud fractional derivative of order $\gamma$ \\
		$[A,B]_\gamma$ & $\gamma$-deformed commutator $AB - e^{i\pi\gamma}BA$ \\
		$a_j(\Phi,\rho), b_j(\Phi,\rho), c_j(\Phi,\rho)$ & Coefficients in the asymptotic expansion \\
		$R_{m,n}(\Phi,\rho)$ & Remainder term in the expansion \\
		$C_{m,\gamma,d}$ & Explicit constant in the remainder estimate \\
		$\Gamma(z)$ & Gamma function \\
		$\binom{j}{\alpha}$ & Multinomial coefficient $\frac{j!}{\alpha_1! \cdots \alpha_d!}$ \\
		$\binom{j}{\alpha,\beta}$ & Multinomial coefficient $\frac{j!}{\alpha!\beta!}$ \\
		$\braket{A,B}$ & Hilbert-Schmidt inner product $\tr(A^*B)$ \\
		$\norm{\bullet}_1$ & Trace norm (nuclear norm) \\
	\end{longtable}
}


\begin{thebibliography}{99}
		
		\bibitem{Voronovskaja1932}
		Voronovskaja, E. (1932). Détermination de la forme asymptotique d'approximation des fonctions par les polynômes de M. Bernstein. \textit{CR Acad. Sci. URSS}, 79, 79-85.
		
		\bibitem{Anastassiou2023}
		Anastassiou, G. A. (2023). \textit{Parametrized, deformed and general neural networks}. Berlin/Heidelberg, Germany: Springer. \url{10.1007/978-3-031-43021-3}
		
		\bibitem{Holevo2019}
		Holevo, A. S. (2019). Quantum Systems, Channels. \textit{In Information}. De Gruyter.
		
		\bibitem{KuboAndo1980}
		Kubo, F., \& Ando, T. (1980). Means of positive linear operators. \textit{Mathematische Annalen, 246(3)}, 205-224. \url{10.1007/BF01371042}.
		
		\bibitem{NielsenChuang2010}
		Nielsen, M. A., \& Chuang, I. L. (2010). \textit{Quantum computation and quantum information}. Cambridge university press.
		
		\bibitem{Paulsen2002}
		Paulsen, V. (2002). \textit{Completely bounded maps and operator algebras} (Vol. 78). Cambridge University Press.
		
		\bibitem{Samko1993}
		Samko, S. G. (1993). Fractional integrals and derivatives. \textit{Theory and applications}.
		
		\bibitem{SteinWeiss1971}
		Stein, E. M., \& Weiss, G. (1971). \textit{Introduction to Fourier analysis on Euclidean spaces} (Vol. 1). Princeton university press.
		
		\bibitem{AmariNagaoka2000}
		Amari, S. I., \& Nagaoka, H. (2000). \textit{Methods of information geometry} (Vol. 191). American Mathematical Soc.
		
	\end{thebibliography}
\end{document}